\documentclass{amsart}
\usepackage{amsmath}
\advance\textheight1.5\baselineskip

\theoremstyle{plain}
\newtheorem{theorem}{Theorem}

\newtheorem{proposition}[theorem]{Proposition}
\theoremstyle{remark}
\newtheorem*{remark*}{Remark}

\newcommand\CC{{\mathbb C}}
\newcommand\RR{{\mathbb R}}
\newcommand\into{\int_\Omega}
\renewcommand\[{\begin{equation}}
\renewcommand\]{\end{equation}}
\newcommand\spr[1]{\langle#1\rangle}

\newcommand\hol{_{\text{hol}}}
\newcommand\ldvah{L^2_h}
\newcommand\ldvahh{L^2_{\text{hol},h}}
\newcommand\LL{{\mathcal L}}
\newcommand\HH{\mathcal H}
\newcommand\HHe{{\HH _\epsilon}}
\newcommand\intR{\int_\RR}
\renewcommand\Re{\operatorname{Re}}
\newcommand\tHHe{\widetilde{\HH}_\epsilon}
\newcommand\tKe{\widetilde K_\epsilon}
\newcommand\MMe{\mathcal M_\epsilon}
\newcommand\tKey{\widetilde K_{\epsilon,y}}
\newcommand\cD{\mathcal D}
\newcommand\WW{\mathcal W}
\newcommand\Kez{K_{\epsilon,z}}
\newcommand\Te{T^{(\epsilon)}}
\newcommand\tTe{\widetilde T^{(\epsilon)}}
\newcommand\e{\mathbf e}
\renewcommand\SS{\mathcal S}

\newcommand\oy{\overline y}
\newcommand\oz{\overline z}

\newcommand\FFe{\mathcal F_\epsilon}
\newcommand\intC{\int_{\CC}}

\begin{document}

\title[Hermite polynomials]{Hermite Polynomials and
 Quasi-classical Asymptotics}
\author{S.~Twareque Ali}
\address{Department of Mathematics and Statistics,
 Concordia University, Montr\'eal, Qu\'ebec, Canada~H3G~1M8}
\email{twareque.ali{@}concordia.ca}
\author{Miroslav~Engli\v s}
\address{Mathematics Institute, Silesian University in Opava,
 Na~Rybn\'\i\v cku~1, 74601~Opava, Czech Republic {\rm and}
 Mathematics Institute, \v Zitn\' a 25, 11567~Prague~1, Czech Republic}
\email{englis{@}math.cas.cz}
\thanks{Research supported by GA\v CR grant no.~201/09/0473 and RVO funding
 for I\v CO~67985840 and Natural Sciences and Engineering Research Council
(NSERC) of Canada.}
\begin{abstract}
We study an unorthodox variant of the Berezin-Toeplitz type of quantization
scheme, on a reproducing kernel Hilbert space generated by the real Hermite
polynomials and work out the associated semi-classical asymptotics. 
\end{abstract}
\maketitle

\section{Introduction}\label{sec1}
At~the heart of most approaches to quantization lies the idea of assigning
to functions~$f$ (the~classical observables) suitable operators~$T_f$ (quantum
observables) depending on an auxiliary parameter $h$ (the~Planck constant)
in~such a way that as $h\searrow0$, $T_f$~possesses an appropriate asymptotic
behaviour reflecting the ``(semi)classical limit'' of~the quantum system.
Typically, the~functions $f$ live on a manifold equipped with symplectic
structure (the~phase space) and the required asymptotic behaviour takes the
form of the ``correspondence principle''
\[ T_f T_g - T_g T_f \approx \frac{ih}{2\pi} T_{\{f,g\}}
 \label{tTA}  \]
where $\{\cdot,\cdot\}$ denotes the Poisson bracket.

For~complex manifolds which are not only symplectic but K\"ahler, a~notable
example of such a construction is the Berezin-Toeplitz quantization,
first formally introduced in \cite{BMS}, though some ideas go back to
Berezin~\cite{BeQ} and similar quantization techniques had also been
introduced by other authors \cite{alidoeb,prugov}. 
Namely, assume for simplicity that the phase space $\Omega$ is simply
connected, so that the K\"ahler form $\omega$ admits a global real-valued
potential~$\Psi$, i.e.~$\omega=\partial\overline\partial\Psi$.
Consider the $L^2$~space
\[ \ldvah = \{f\text{ measurable on }\Omega: \into |f|^2 e^{-\Psi/h}
 \,\omega^n <\infty \} \qquad(h>0),  \label{tTB}   \]
and let $\ldvahh$ (the~weighted Bergman space) be~the subspace in $L^2_h$
of functions holomorphic on $\Omega$, and $P_h:\ldvah\to\ldvahh$ the orthogonal
projection. For~a~bounded measurable function $f$ on~$\Omega$, one~then defines
the Toeplitz operator $T_f$ on $\ldvahh$ with symbol $f$~by
\[ T_f u = P_h(fu).   \label{tTC}   \]
This~is, in~fact, an~integral operator: more precisely, the space $\ldvahh$
turns out to be a reproducing kernel Hilbert space~\cite{Aro} possessing a
reproducing kernel $K_h(x,y)$, and (\ref{tTC}) can be rewritten~as
\[ T_f u (x) = \into u(y) f(y) K_h(x,y) \, e^{-\Psi(y)/h} \,\omega(y)^n .
 \label{tTD}  \]
When the manifold $\Omega$ is not simply connected, one has to assume that
the cohomology class of $\omega$ is integral, so~that there exists a Hermitian
line bundle $\LL$ with the canonical connection whose curvature form coincides
with~$\omega$; and the spaces $\ldvahh$ (and~$\ldvah$) get replaced by the
space of all holomorphic (or~all measurable, respectively) square-integrable
sections of~$\LL^{\otimes k}$, $k=\frac1h=1,2,3,\dots$. In~any case, under
reasonable technical assumptions on~$\Omega$ (e.g.~for $\Omega$
compact~\cite{BMS}, or for $\Omega$ simply connected strictly-pseudoconvex
domain in $\CC^n$ with smooth boundary $\partial\Omega$ and $e^{-\Psi}$
vanishing to exactly first order at~$\partial\Omega$~\cite{Ecmp}),
the~Toeplitz operators satisfy
\[ T_f T_g \approx T_{fg}+h T_{C_1(f,g)} + h^2 T_{C_2(f,g)} + \dots \qquad
 \text{as } h\searrow0,   \label{tTE}   \]
with some bidifferential operators $C_j$ such that $C_1(f,g)-C_1(g,f)=\frac
i{2\pi}\{f,g\}$, implying in particular that (\ref{tTA}) holds.
The~asymptotic expansion (\ref{tTE}) even holds in the strongest possible
sense of operator norms, i.e.~the difference of the left-hand side and the
sum of the first $N$ terms on the right hand side has norm, as~an operator
on~$\ldvahh$, bounded by a multiple of $h^N$ as $h\searrow0$, for all $N=1,
2,3,\dots$. Furthermore, the bidifferential operators $C_j$ can be expressed
in terms of covariant derivatives, with contractions of the curvature tensor
and its covariant derivatives as coefficients, thus encoding various geometric
properties of $(\Omega,\omega)$ in an intriguing~way.

The~Berezin-Toeplitz \emph{Ansatz} above has subsequently been extended
to a number of more general contexts outside the K\"ahler setting,
including e.g.~that of harmonic Bergman spaces on some special
domains \cite{EFock} \cite{Blaschke}~\cite{Jahn}, or~when spaces of
holomorphic functions/sections are replaced by eigenspaces of the
$\operatorname{Spin}^c$-Dirac operator on a general symplectic manifold
or even orbifold \cite{DaiLiuMa} \cite{MaMa} \cite{ShiffZeld}~\cite{BwUribe},
while numerous other developments concerned the properties of the cochains
$C_j$ or miscellaneous representation-theoretic aspects of the procedure
\cite{KarbSchl} 
\cite{ReshTakh} 
\cite{Gammelg} 
\cite{DougKle} 
\cite{HXua}~\cite{HXub} 
\cite{AUrealcov}. 

The~purpose of the present paper is to highlight an operator calculus of
a completely different flavour, which nonetheless bears certain resemblance
to~(\ref{tTA}) and~(\ref{tTD}), and arises in a quite unexpected setting ---
namely, in~connection with orthogonal polynomials. Generically, the situation
is the following: as explained above, the Berezin-Toeplitz type of quantization
relies on the existence of a certain $L^2$-space which contains a reproducing
kernel Hilbert space as a subspace; the quantization is effected by the
projection operator of this subspace. Alternatively, the reproducing kernel
$K(x,y)$ defines a family of vectors $K(\cdot , y), \; y \in \Omega$ in the
reproducing kernel Hilbert space, generally called {\em coherent states} in the
literature, and then (\ref{tTD}) shows that the quantization may also be
defined in terms of these coherent states. However, the existence of coherent
states depends only on the reproducing kernel and not on any ambient
$L^2$-space and indeed, there have been proposals, some very recent
\cite{horszaf,odzhor}, to base both the theory of coherent states and geometric
quantization using a positive definite kernel alone. The present paper may be
thought of as an extension of this line of thought to Berezin-Toeplitz
quantization. What is interesting in our present case is that it is the Hermite
polynomials, which in a way define the quantum harmonic oscillator, also define
the reproducing kernel of our problem. 

To~be more specific, let $H_n(x)$ stand for the standard Hermite polynomials
(see~Section~\ref{sec2} below for the details), and, for $0<\epsilon<1$, set
\[ K_\epsilon(x,y) = \sum_{n=0}^\infty \epsilon^n \|H_n\|^{-2} H_n(x)
 \overline{H_n(y)}, \qquad x,y\in\RR.   \label{tTF}  \]
Here $\|H_n\|$ denotes the norm in $L^2(\RR,e^{-x^2}\,dx)$, where the $\{H_n\}$
form an orthogonal basis. Then $K_\epsilon$ is a positive-definite function,
and, hence, determines uniquely a Hilbert space $\HHe$ of functions on $\RR$
for which $K_\epsilon$ is the reproducing kernel~\cite{Aro}; this space first
appeared in~\cite{AliKr} when studying ``squeezed'' coherent
states. (Its~definition may perhaps seem a bit artificial at first glance,
but~so must have seemed (\ref{tTB}) when it first came around in Berezin's
papers!) For~a (reasonable) function $f$ on~$\RR$, set
\[ T_f u(x) := \intR u(y) f(y) K_\epsilon(x,y) \, e^{-y^2}\,dy.  \label{tTG} \]
This certainly resembles the expression (\ref{tTD}) for Toeplitz operators,
however, note that this time there is no $L^2$ space around like~(\ref{tTB})
which would contain $\HHe$ as a closed subspace (in~fact, the set
$\{f(x)e^{-x^2/2}: f\in\HHe\}$ is a dense, rather than proper closed,
subset of~$L^2(\RR)$), so~there is no projection like $P_h$ around and
the original definition (\ref{tTC}) makes no sense.
In~particular, there is no reason \emph{a priori} even to expect (\ref{tTG})
to be defined, not to say bounded, on~some space (whereas with (\ref{tTC}) it
immediately follows that $\|T_f\|$ is not greater than the norm of the operator
of ``multiplication by~$f$'' on~$L^2$, hence $\|T_f\|\le\|f\|_\infty$).
It~may therefore come as a bit of a surprise that
(\ref{tTG}) actually yields, for $f\in L^\infty(\RR)$, a~bounded operator
on~$L^2(\RR)$, and, moreover, $T_f$~enjoys a nice asymptotic behaviour as
$\epsilon\nearrow1$, which we will see to correspond, in~a~very natural sense,
to~the semiclassical limit $h\searrow0$ in the original quantization setting.

It~should be stressed that the resulting asymptotics are not quite of the
form (\ref{tTE}) and, in~particular, (\ref{tTA})~does not hold, so~that our
results claim no direct physical relevance; on~the other hand, the~same
is true as well for some of the generalizations of the classical Toeplitz
calculus mentioned two paragraphs above, while not depriving the latter of
their mathematical beauty and relevance. We~hope the same to be at least
partly true also for our developments here and thus justify their disclosure
to a wider audience.

We~review the necessary standard material on Hermite polynomials in
Section~\ref{sec2}. The~spaces $\HHe$ are discussed in Section~\ref{sec3},
and the basic facts about the operators $T_f$ from (\ref{tTG}) in
Section~\ref{sec4}. The~asymptotic behaviour is studied in Section~\ref{sec5}.
In~Section~\ref{secQ} we observe how to recover the standard Berezin-Toeplitz
quantization on $\CC$ using the Hermite {\it Ansatz} and an appropriate
analogue of the Bargmann transform.

A large portion of this work was done while the second author was visiting the
first in September~2012; the hospitality of the Department of Mathematics and
Statistics of Concordia University on that occasion is gratefully acknowledged.

\section{Hermite polynomials} \label{sec2}
The Hermite polynomials $H_n(x)$, $n=0,1,2,\dots$, are defined by the formula
\[ H_n(x) := (-1)^n e^{x^2} \frac{d^n}{dx^n} e^{-x^2}.   \label{tHA}  \]
They can also be obtained from the generating function
\[ e^{2xz-z^2} = \sum_{n=0}^\infty \frac{z^n}{n!} H_n(x)  \label{tHB}  \]
and satisfy the orthogonality relations
\[ \intR H_n(x) H_m(x) \, e^{-x^2} \,dx = n!2^n\sqrt\pi \delta_{mn}.
 \label{tHC}  \]
It~follows that the \emph{Hermite functions}
\[ h_n(x) := (n!2^n\sqrt\pi)^{-1/2} H_n(x) e^{-x^2/2},  \label{tHD}  \]
$n=0,1,2,\dots$, form an orthonormal basis in the Hilbert space $L^2(\RR)$
on the real line.

The~representation (\ref{tHB}) also leads to the explicit formula
\[ H_n(x) = n! \sum_{m=0}^{[n/2]} \frac{(-1)^m(2z)^{n-2m}} {m!(n-2m)!},
 \label{tHE}   \]
$[x]$~denoting the integer part of~$x$. From this follows the estimate
\[ |H_n(z)| \le \sqrt{n!2^n} e^{\sqrt{2n}|z|}    \label{tHF}   \]
valid for all complex~$z$.

All~this, of~course, is~quite standard and well-known (see~e.g.~\cite{Askey},
Chapter~6.1), perhaps with the exception of the last estimate for non-real~$z$;
for~completeness, we~therefore attach a proof. Observe first of all that
\[ \sqrt{\frac{n!}{n^n}} \le \frac{2^{[n/2]}[\frac n2]!} {n^{[n/2]}}.
 \label{tHG}   \]
Indeed, both for $n=2k$ even and for $n=2k+1$ odd, the corresponding
inequalities
$$ \frac{\sqrt{(2k)!}}{(2k)^k} \le \frac{2^kk!}{(2k)^k}, \qquad
 \frac{\sqrt{(2k+1)!}}{(2k+1)^{k+\frac12}} \le \frac{2^kk!}{(2k+1)^k}  $$
reduce to the elementary estimate
$$ (2k)! \le (2\cdot4\cdot6\dots\cdot(2k))^2 = 4^k k!^2.   $$
Since
$$ \frac{2^mm!}{n^m} = \prod_{j=1}^m \frac j{n/2}   $$
is a decreasing function of $m$ for $0\le m\le\frac n2$, it~follows from
(\ref{tHG}) that even
$$ \sqrt{\frac{n!}{n^n}} \le \frac{2^mm!}{n^m}, \qquad
 m=0,1,\dots,\Big[\frac n2\Big].   $$
Consequently,
$$ \frac{\sqrt{n!}}{m!} 2^{\frac n2-2m} \le n^{\frac n2-m} 2^{\frac n2-m}
 = (2n)^{\frac{n-2m}2}   $$
and
$$ \sum_{m=0}^{[n/2]} \frac{\sqrt{n!} 2^{\frac n2-2m} |z|^{n-2m}} {m!(n-2m)!}
 \le \sum_{m=0}^{[n/2]} \frac{(\sqrt{2n}|z|)^{n-2m}}{(n-2m)!}
 \le e^{\sqrt{2n}|z|},  $$
proving (\ref{tHF}).

As~a~corollary, we~also get the estimate
\[ |h_n(z)| \le \pi^{-1/4} e^{\sqrt{2n}|z|-\Re z^2/2},  \qquad z\in\CC,
 \label{tHH}  \]
for the corresponding Hermite functions.

The~last fact we need to recall is the differential equation
$$ H_n''(x) - 2x H'_n(x) + 2n H_n(x) = 0  $$
for $H_n(x)$, which translates into another differential equation
$$ h''_n(x) + (2n+1-x^2) h_n(x) = 0  $$
for the Hermite functions~$h_n$. In~other words, the (Schr\"odinger) operator
\[ A:= \frac{x^2-1}2I - \frac12 \frac{d^2}{dx^2}  \label{tHI}   \]
on $L^2(\RR)$ satisfies
\[ A h_n = n h_n, \qquad n=0,1,2,\dots  \label{tHJ}  \]
(that~is, $A=\sum_n n \spr{\cdot,h_n}h_n$).

\section{Reproducing kernel spaces} \label{sec3}
For $0<\epsilon<1$, the reproducing kernels
\[ K_\epsilon(x,y) = \sum_{n=0}^\infty \epsilon^n \frac{H_n(x)H_n(y)}
 {n!2^n\sqrt\pi} = \frac1{\sqrt{(1-\epsilon^2)\pi}} e^{-\frac{\epsilon^2}
 {1-\epsilon^2}(x^2+y^2-\frac2\epsilon xy)} ,  \label{tSA}   \]
were introduced in~\cite{AliKr}; the second equality is known as Mehler's
formula. We~denote by $\HHe$ the corresponding reproducing kernel
space~\cite{Aro}; that~is, $\HHe$~is the completion of linear combinations
of the functions $K_\epsilon(\cdot,y)$, $y\in\RR$, with respect to the
scalar product
$$ \Big\langle\sum_j a_j K_\epsilon(\cdot,y_j),\sum_k b_k K_\epsilon(\cdot,x_k)
 \Big\rangle = \sum_{j,k} a_j\overline{b_k} K_\epsilon(x_k,y_j).   $$
We~will also use the Hilbert spaces
$$ \tHHe = e^{-x^2/2}\HHe = \{e^{-x^2/2}f(x): f\in\HHe \}  $$
corresponding to the reproducing kernel
\[ \tKe(x,y) = e^{-(x^2+y^2)/2} K_\epsilon(x,y) = \sum_{n=0}^\infty
 \epsilon^n h_n(x) h_n(y).   \label{tSK}   \]
Since the transition from $\HHe$ to $\tHHe$ involves only the multiplication
by~$e^{-x^2/2}$, we~will state the various facts below usually only for one
of these spaces.

The~following assertion, though not explicitly stated in~\cite{AliKr},
is~fairly straightforward.

\begin{proposition} \label{Prop1}
One has
\[ \tHHe = \{f(x)=\sum_n f_n h_n(x): \sum_n \epsilon^{-n}|f_n|^2<\infty \}
 \label{tSB}  \]
with the norm in $\tHHe$ being given~by
\[ \|f\|_\epsilon^2 = \sum_n \epsilon^{-n} |f_n|^2.   \label{tSC}  \]
\end{proposition}

\begin{proof} Let~us temporarily denote the space on the right-hand side of
(\ref{tSB}) (with the norm given by~(\ref{tSC})) by~$\MMe$. From the equality
$$ \sum_n |\epsilon^n h_n(y)|^2 \epsilon^{-n} = \sum_n \epsilon^n|h_n(y)|^2
 = \tKe(y,y) = \frac1{\sqrt{(1-\epsilon^2)\pi}}
 e^{\frac{\epsilon-1}{\epsilon+1}x^2} < \infty   $$
(cf.~(\ref{tSA})), it~follows that the function
$$ \tKey:= \sum_n \epsilon^n h_n(y) h_n = \tKe(\cdot,y)   $$
belongs to~$\MMe$, for any $y\in\RR$. Furthermore, for any $f=\sum_n f_n h_n
\in\HHe$,
$$ \spr{f,\tKey}_\epsilon = \sum_n \epsilon^{-n} f_n \overline{\epsilon^n
 h_n(y)} = \sum_n f_n h_n(y) = f(y)  $$
(here we have used the fact that $h_n$ is real-valued on~$\RR$). Thus~$\tKe$
is the reproducing kernel for~$\MMe$. Since a reproducing kernel Hilbert space
is uniquely determined by its reproducing kernel, $\MMe=\tHHe$, with equality
of norms.   \end{proof}

The~last proposition allows for the following interpretation of the spaces~$\tHHe$
and~$\HHe$. Recall that the Sobolev space of order $s$ on $\RR$ can be defined
as the (completion of~the) space of all~$f$ ($\in\cD(\RR)$) for which
$$ \|f\|^2_s := \spr{(I-\Delta)^s f,f}_{L^2(\RR)} < \infty.  $$
By~analogy, one~could define ``Hermite-Sobolev'' spaces $\WW^s(\RR)$ on
$\RR$~by
$$ \|f\|_s^2 := \spr{(I+A)^s f,f}_{L^2(\RR)} < \infty.   $$
In~view of~(\ref{tHJ}), this is equivalent~to
$$ \WW^s(\RR) = \{f=\sum_n f_nh_n: \|f\|_s^2=\sum_n (n+1)^s|f_n|^2<\infty\}. $$
Our~spaces $\tHHe$ are thus obtained upon replacing $(n+1)^s$
by~$\epsilon^{-n}$. Back in the context of the ordinary Laplacian, they are
thus analogues of the spaces
$$ e^{\epsilon\Delta/2}L^2(\RR) = \{f: \spr{e^{-\epsilon\Delta}f,f}
 _{L^2(\RR)} < \infty \}  $$
of~solutions at time $t=\frac\epsilon2$ of the heat equation
$\frac{\partial u}{\partial t}=\Delta u$, $u(x,0)=f(x)$
(``caloric functions''). More precisely, $\tHHe=e^{-A\log\sqrt\epsilon}
L^2(\RR)$ is~the space of solutions at time $t=-\frac12\log\epsilon$ of
the modified heat equation
$$ \frac{\partial u}{\partial t} = Au, \qquad u=u(x,t), \; t>0,   $$
with initial condition $u(\cdot,0)\in L^2(\RR)$.

We~conclude this section by showing that $\HHe$ is actually a space of
holomorphic functions, like the weighted Bergman spaces mentioned in
the Introduction.
(The~same is true also for the ordinary spaces of caloric functions.)

\begin{theorem} \label{thm2}
Each $f\in\HHe$ extends to an entire function on~$\CC$, and $\HHe$ is the
space of $($the~restrictions to~$\RR$~of$)$ holomorphic functions on $\CC$
with reproducing kernel
\[ K_\epsilon(x,y) = \sum_{n=0}^\infty \epsilon^n
 \frac{H_n(x)\overline{H_n(y)}} {n!2^n\sqrt\pi}
 = \frac{e^{-\frac{\epsilon^2}{1-\epsilon^2}(x^2+\oy^2-\frac2\epsilon x\oy)}}
 {\sqrt{(1-\epsilon^2)\pi}} .   \label{tSE}  \]
\end{theorem}

\begin{proof} By~the preceding proposition, we~have
\[ f = \sum_n f_n (n!2^n\sqrt\pi)^{-1/2} H_n,  \label{tSD}  \]
with
$$ \sum_n \epsilon^{-n} |f_n|^2 = \|f\|^2_{\HHe} < \infty .  $$
Consequently, for any $z\in\CC$, we~get using the estimate~(\ref{tHF})
\begin{align*}
 \sum_n |f_n H_n(z) (n!2^n\sqrt\pi)^{-1/2}|
 &\le \sum_n |f_n| \pi^{-1/4} e^{\sqrt{2n}|z|} \\
 &\le \pi^{-1/4} \|f\|_{\HHe} \Big(\sum_n \epsilon^n
 |e^{\sqrt{2n}|z|}|^2 \Big)^{1/2}.   \end{align*}
Since for any fixed $z\in\CC$, the radius of convergence of $\sum_n \epsilon^n
e^{2\sqrt{2n}|z|}$, with $\epsilon$ as the variable, is~1, the~expression in
the last parentheses is finite for $0<\epsilon<1$. Thus the series~(\ref{tSD})
converges for any $z\in\CC$ (and~uniformly on compact subsets). This proves the
first part of the theorem, and also shows that
$$ f(z) = \spr{f,\Kez}_{\HHe}, \qquad z\in\CC,   $$
with
$$ \Kez := \sum_n \epsilon^n \overline{H_n(z)(n!2^n\sqrt\pi)^{-1/2}}
 (n!2^n\sqrt\pi)^{-1/2} H_n,   $$
that is,
$$ \Kez(w) = \sum_n \epsilon^n\frac{H_n(w)\overline{H_n(z)}}{n!2^n\sqrt\pi}, $$
showing that (\ref{tSE}) is indeed the reproducing kernel for $\HHe$ on all
of~$\CC$.  \end{proof}

\section{Toeplitz-type operators} \label{sec4}
Drawing inspiration from~(\ref{tTD}), we~define, for a function (``symbol'')
$f$~on~$\RR$, the ``Toeplitz operator'' $\Te_f$, $0<\epsilon<1$,
on~$L^2(\RR)$~by
\[ \Te_f u(x) := \intR u(y) f(y) K_\epsilon(x,y) \, e^{-y^2}\,dy.
 \label{tRA}   \]
We~will also use the analogous operators
\[ \begin{aligned}
 \tTe_f u(x) :&= \intR u(y) f(y) \tKe(x,y) \, dy   \\
 &= \intR u(y) f(y) K_\epsilon(x,y) \, e^{-\frac{x^2+y^2}2} \, dy
 \end{aligned}  \label{tRB}   \]
on $L^2(\RR)$ defined using the kernel $\tKe$ instead of~$K_\epsilon$. Clearly,
\[ \tTe_f u = \e^{-1/2} \Te_f \e^{1/2}   \label{tRC}   \]
where we introduced the notation
$$ \e(x) := e^{x^2}.   $$
It~turns out that the operators $\Te$ have a bit nicer expression in terms of
the Fourier transform, while $\tTe$ are a bit nicer from the point of view of
the ``semiclassical'' asymptotics as $\epsilon\nearrow1$.
In~view of~(\ref{tRC}), it~is always a simple matter to pass from $\Te$ to
$\tTe$ or vice versa.

In~the formula~(\ref{tTD}), the reproducing kernel $K_h(x,y)$ is the integral
kernel of the orthogonal projection $P_h$ onto~$\ldvahh$, i.e.~of~a~bounded
operator in the corresponding space~$\ldvah$. On~the other hand, for~$\tKe$
we have no such interpretation, in~fact the space~$\tHHe$, of~which $\tKe$
is the reproducing kernel, is~dense in~$L^2(\RR)$
(this is immediate from Proposition~\ref{Prop1} and the fact that $\{h_n\}
_{n=0}^\infty$ is an orthonormal basis of~$L^2(\RR)$). The~next two results
may therefore seem somewhat surprising.

\begin{theorem}
The operators $\Te_f$ and $\tTe_f$ are densely defined for any
$f\in C^\infty(\RR)$. Furthermore, $\Te_f$~is bounded for $f\in L^\infty(\RR)$,
with
$$ \|\Te_f\| \le C_\epsilon \|f\|_\infty   $$
for some constant $C_\epsilon$ depending only on~$\epsilon$, $0<\epsilon<1$.
\end{theorem}

\begin{proof} By~(\ref{tSA}),
\[ \begin{aligned}
\Te_f u(x) &= \intR (fu)(y) e^{-\frac{\epsilon^2}{1-\epsilon^2}
 (x^2-\frac 2\epsilon xy+y^2)-y^2} \, \frac{dy}{\sqrt{(1-\epsilon^2)\pi}} \\
&= \intR (fu)(y) e^{-\frac{\epsilon^2}{1-\epsilon^2}
 (x-\frac y\epsilon)^2} \, \frac{dy}{\sqrt{(1-\epsilon^2)\pi}}  \\
&= \intR (fu)(\epsilon x-\sqrt{1-\epsilon^2}t) \;
 e^{-t^2} \frac{dt}{\sqrt\pi} \\
&= (\delta_{\sqrt{1-\epsilon^2}}(fu)*\e^{-1})
 \Big(\frac{\epsilon x}{\sqrt{1-\epsilon^2}}\Big) ,  \end{aligned}
 \label{tRF}    \]
where we have introduced the dilation operator
$$ \delta_r u(x) := u(rx) .   $$
In~other words, introducing also the operator
$$ G u:= u * \e^{-1}   $$
of convolution with the Gaussian~$\e^{-1}$, we~obtain
\[ \Te_f = \delta_{\epsilon/\sqrt{1-\epsilon^2}} G
 \delta_{\sqrt{1-\epsilon^2}} M_f,  \label{tRD}   \]
where
$$ M_f: u \mapsto fu   $$
denotes the operator of ``multiplication by~$f$''. If~$f\in C^\infty(\RR)$ and
$u\in\cD(\RR)$, the space of smooth functions on~$\RR$ with compact support,
then $fu\in\cD\subset\SS$, the Schwartz space on~$\RR$. Since dilations map
$\SS$ into itself while
\[ Gf = \Big( \hat f \, \frac{\e^{-1/4}}{2\sqrt\pi} \Big) ^\vee  \label{tRE} \]
(here $\hat{\;}$ and ${}^\vee$ denote the Fourier transform and the inverse
Fourier transform, respectively) also maps $\SS$ into itself, we~conclude that
$$ \Te_f u \in\SS \qquad\text{for any } f\in C^\infty(\RR) \text{ and }
 u\in\cD(\RR).  $$
Since $\cD$ is dense in $L^2$ and $\SS\subset L^2$, this proves the first part
of the theorem for~$\Te$. The~assertion for $\tTe$ is then immediate from
(\ref{tRC}) and the fact that $\e^{1/2}\cD\subset\cD$ and $\e^{-1/2}L^2
\subset L^2$.

The~second part follows from~(\ref{tRD}), because $\|M_f\|\le\|f\|_\infty$ and
$$ \|\delta_{\epsilon/\sqrt{1-\epsilon^2}} G \delta_{\sqrt{1-\epsilon^2}} \|
 = (4\pi\epsilon)^{-1/2} =: C_\epsilon < \infty  $$
by~an elementary argument and standard properties of the Fourier transform.
\end{proof}

\begin{theorem}
For $f\in L^\infty$ the operator $\tTe_f$ is bounded on~$L^2(\RR)$.
\end{theorem}

\begin{proof} By~(\ref{tSK})
$$ \tTe_f u = \sum_n \epsilon^n \spr{fu,h_n}h_n.   $$
Thus, for any $0<\epsilon<1$,
$$ \|\tTe_f u\|^2 = \sum_n \epsilon^{2n} |\spr{fu,h_n}|^2
 \le \sum_n |\spr{fu,h_n}|^2 = \|fu\|^2 \le \|f\|_\infty^2 \|u\|^2,  $$
so $\|\tTe_f\|\le\|f\|_\infty$.  \end{proof}

We~remark that the same argument as in the last proof also shows that $\Te_f$
is~bounded, for any $f\in L^\infty$, on~the weighted space
$L^2(\RR,e^{-x^2}\,dx)$.

\section{``Semiclassical'' asymptotics} \label{sec5}
The~Parseval identity
$$ f = \sum_n \spr{f,h_n} h_n ,  \qquad f\in L^2(\RR),   $$
shows that, at~least in the weak sense
(i.e.~as distributions on~$\RR\times\RR$),
\[ \sum_n h_n(x) h_n(y) = \delta(x-y) .   \label{tRG}   \]
Thus formally
$$ \Te_f u = fu  \qquad\text{for } \epsilon=1,   $$
that~is, the operator $\Te_f$ reduces just to the multiplication operator
$M_f$ on $L^2(\RR)$ (in~the sense explained above) for $\epsilon=1$.
This brings forth naturally the question of the finer description of the
behaviour of $\Te_f$ as $\epsilon\nearrow1$, in~particular, whether one has
any analogue of the ``semiclassical limit'' formulas like (\ref{tTA}) or
(\ref{tTE}) in the traditional procedures.

The~latter asymptotics can be found by the usual Laplace (or~stationary
phase, or~WJKB) method, see~e.g.~H\"ormander~\cite[\S7.7]{Hrm}.
Namely, assume for simplicity that $f\in C^\infty(\RR)$ and $u\in\cD(\RR)$.
We~have seen in (\ref{tRF}) that
\begin{align*}
\Te_f u(x) &= \intR (fu)(y) e^{-\frac{(y-\epsilon x)^2}{1-\epsilon^2}} \,
 \frac{dy}{\sqrt{(1-\epsilon^2)\pi}}   \\
 &= \intR (fu) (\epsilon x-\sqrt{1-\epsilon^2}t) \; e^{-t^2}
 \, \frac{dt}{\sqrt\pi} .   \end{align*}
Let~us temporarily write, for the sake of brevity, $fu=F$. Standard estimates
used in the stationary phase method show that the integration over $y$ outside
a small neighbourhood of $x$ gives an exponentially small contribution as
$\epsilon\nearrow1$, while in the integral over that neighbourhood $F$ can
be replaced by its Taylor expansion. Thus we arrive~at
\begin{align*}
\intR F(\epsilon x-\sqrt{1-\epsilon^2}t) \; e^{-t^2} \,\frac{dt}{\sqrt\pi}
 &\approx \sum_{k=0}^\infty \frac{F^{(k)}(x)}{k!} \intR
 (\epsilon x-\sqrt{1-\epsilon^2}t)^k e^{-t^2} \,\frac{dt}{\sqrt\pi}  \\
 &= \sum_{j,l=0}^\infty \frac{F^{(j+l)}(x)}{j!l!} (\epsilon-1)^l x^l
 (-\sqrt{1-\epsilon^2})^j \intR t^j e^{-t^2} \,\frac{dt}{\sqrt\pi}  \\
 &= \sum_{k,l=0}^\infty \frac{F^{(2k+l)}(x)}{(2k)!l!} (\epsilon-1)^l x^l
 (1-\epsilon^2)^k \frac{\Gamma(k+\frac12)}{\Gamma(\frac12)}
\end{align*}
as~$\epsilon\nearrow1$. Writing $1-\epsilon^2=(1-\epsilon)(2-(1-\epsilon))$
and using the binomial theorem to get powers of $(1-\epsilon)$ only,
we~finally~get
\[ \Te_f u(x) \approx \sum_{k,l,m=0}^\infty (1-\epsilon)^{k+l+m}
 \frac {(fu)^{(2k+l)}(x) x^l (-1)^{l+m} 2^{k-m} \binom km} {l!k!4^k}
 \label{tRH}   \]
as~$\epsilon\nearrow1$. In~particular,
\[ \Te_f u = fu + (1-\epsilon) \Big[ \Big( \frac{f''}2-xf'\Big) u + (f'-xf)u'
 + \frac f2 u'' \Big] + O((1-\epsilon)^2).  \label{tRI}   \]

A~similar approach could, of course, be~applied also to~$\tTe_f$; however,
we~proceed to use a different argument, which not only recovers the
formula~(\ref{tRH}) (upon~passing from $\tTe$ to $\Te$ via the
relation~(\ref{tRC})) but~is also shorter and applicable in other situations.

Recall the Schr\"odinger (``number'') operator
$$ A = \frac{x^2-1}2 I - \frac12 \frac{d^2}{dx^2}   $$
which is an (unbounded) self-adjoint operator on $L^2(\RR)$ satisfying
$Ah_n=nh_n$, $n=0,1,2,\dots$.

\begin{theorem}
We~have
$$ \tTe_f = \epsilon^A M_f,   $$
where $\epsilon^A=e^{A\log\epsilon}$ is understood in the sense of the
spectral theorem. Consequently, as~$\epsilon\nearrow1$,
\[ \tTe_f u \approx \sum_{k=0}^\infty \frac{(\log\epsilon)^k}{k!} A^k(fu).
 \label{tRJ}   \]
\end{theorem}

\begin{proof} Let~us keep our shorthand $F=uf$, assuming for simplicity that
$F\in\cD(\RR)$. Then
\begin{align*}
\intR F(y) \tKe(x,y) \, dy
&= \intR F(y) \sum_n \epsilon^n h_n(x) h_n(y) \,dy\\
&= \sum_n \epsilon^n \spr{F,h_n} h_n(x)  \\
&= \sum_n \spr{F,h_n} \epsilon^A h_n(x)  \\
&= \Big(\epsilon^A \sum_n \spr{F,h_n} h_n \Big) (x)  \\
&= (\epsilon^A F)(x) = \sum_k \frac{(\log\epsilon)^k}{k!} (A^k F)(x).
\end{align*}
Recalling that $F=fu$ gives the result.   \end{proof}

Of~course, using the familiar series
$$ \log\epsilon = -\sum_{j=1}^\infty \frac{(1-\epsilon)^j}j  $$
one could easily pass in (\ref{tRJ}) from powers of $\log\epsilon$ to powers
of $(1-\epsilon)$.

The~beginning of the asymptotic expansion (\ref{tRJ}) reads
$\tTe_fu=fu+(1-\epsilon)A(fu)+O((1-\epsilon)^2)$, or
\[ \tTe_f = M_f + (1-\epsilon)AM_f + O((1-\epsilon)^2) .  \label{tRK}  \]
Using the similar formulas for $g$ and $fg$ and subtracting, we~arrive~at
\[ \tTe_f \tTe_g - \tTe_{fg} = (1-\epsilon) M_fAM_g + O((1-\epsilon)^2)
 \label{tRL}  \]
and
\begin{align}
\tTe_f\tTe_g-\tTe_g\tTe_f &= (1-\epsilon)(M_fAM_g-M_gAM_f) +O((1-\epsilon)^2)
 \nonumber  \\
&= \frac{1-\epsilon}2 (M_f D^2 M_g-M_g D^2 M_f) + O((1-\epsilon)^2)
 \nonumber  \\
&= (1-\epsilon)(M_{\frac{fg''-gf''}2} + M_{fg'-gf'}D) + O((1-\epsilon)^2),
 \label{tRM}  \end{align}
where we introduced the notation
$$ Du(x) := \frac{du(x)}{dx}   $$
for the differentiation operator on~$\RR$. Comparing these formulas with
(\ref{tTA}) and~(\ref{tTE}) --- the~role of the Planck constant being now
played by the quantity $1-\epsilon$ --- we~see that, first of~all, the~role
of the Poisson bracket is now played by the (second-order) expression
$\dfrac{fg''-gf''}2$; and, secondly, that in addition to the ``Toeplitz''
operators~$\tTe$, the~differentiation operator $D$ appears too.

For~$\Te$ instead of~$\tTe$, the~formulas (\ref{tRL}) and (\ref{tRM}) get
replaced~by
$$ T_fT_g-T_{fg} = (1-\epsilon)\Big[\Big(\frac{fg''}2-xfg'\Big)I + (fg'-xfg)D
 +\frac{fg}2 D^2\Big] + O((1-\epsilon)^2)   $$
and
$$ \postdisplaypenalty1000000
 T_fT_g-T_gT_f = (1-\epsilon)\Big[\Big(\frac{fg''-gf''}2+xf'g-xfg'\Big)I +
 (fg'-f'g)D \Big] + O((1-\epsilon)^2),   $$
respectively, and a similar comment applies.

\section{Berezin-Toeplitz quantization via Hermite polynomials} \label{secQ}
By~virtue of~(\ref{tSE}), the~multiplication operator
$$ M: f(z) \longmapsto \frac{\sqrt{2\epsilon}}{(1-\epsilon^2)^{1/4}\pi^{1/4}}
 e^{\frac{\epsilon^2}{1-\epsilon^2}z^2} f(z)  $$
maps the space $\HHe$ onto the space of holomorphic functions on~$\CC$ with
reproducing kernel
$$ F_\epsilon(z,w) := \frac{2\epsilon}{(1-\epsilon^2)\pi} K_\epsilon(z,w)
 = \frac{2\epsilon}{(1-\epsilon^2)\pi} e^{\frac{2\epsilon}{(1-\epsilon^2)}
 z\overline w} ,  $$
that~is, onto the standard Fock (Segal-Bargmann) space
$$ \FFe = L^2\hol(\CC, d\mu_\epsilon)   $$
of~all entire functions on $\CC$ square-integrable with respect to the
Gaussian measure
$$ d\mu_\epsilon(z) := e^{-2\epsilon|z|^2/(1-\epsilon)}\,dz,   $$
where $dz$ stands for the Lebesgue area measure on~$\CC$.
This can also be checked directly, using the orthogonality relation
\[ \intC H_n(z) \overline{H_m(z)} \; e^{-\frac{2\epsilon}{1+\epsilon} x^2
 -\frac{2\epsilon}{1-\epsilon} y^2} \; dx \, dy =
 \frac{\sqrt{1-\epsilon^2}}{2\epsilon} \, n!2^n\pi \epsilon^{-n} \delta_{mn},
 \qquad z=x+yi,   \label{tVI}  \]
which can be verified using the generating function for~$H_n$, and which
implies that the orthonormal basis $\{\epsilon^{n/2} (n!2^n\sqrt\pi)^{-1/2}
H_n(z)\}_{n=0}^\infty$ of $\HHe$ is (taking $z$ complex) also an orthonormal
basis in $L^2\hol(\CC,\frac{2\epsilon}{\sqrt{(1-\epsilon^2)\pi}}
e^{-\frac{2\epsilon}{1+\epsilon} x^2 -\frac{2\epsilon}{1-\epsilon} y^2}
\,dx\,dy)$; see~\cite{AliKr}.

Correspondingly,
$$ E_n(z) := \frac{\epsilon^{n/2}}{\sqrt{n!2^n\pi^{1/2}}}
 \frac{\sqrt{2\epsilon}}{\root4\of{(1-\epsilon^2)\pi}} H_n(z)
 e^{\epsilon^2 z^2/(1-\epsilon^2)}, \qquad n=0,1,2,\dots,   $$
form an orthonormal basis in~$\FFe$. The~operator
\[ V: f \longmapsto \sum_n \spr{f,h_n}E_n   \label{tVZ}   \]
taking each $h_n$ into $E_n$ is thus a unitary map of $L^2(\RR)$ onto~$\FFe$,
which is a ``Hermite'' analogue of the Bargmann transform. Explicitly,
\[ Vf(z) = \intR f(y) \beta(z,y) \, dy ,  \label{tVA}   \]
where
\begin{align}
\beta(z,y) &= \sum_n h_n(y) E_n(z) \nonumber  \\
&= \frac{\sqrt{2\epsilon}}{\root4\of{(1-\epsilon^2)\pi}}
 e^{\frac{\epsilon^2 z^2}{1-\epsilon^2}-\frac{y^2}2}
 K_{\sqrt\epsilon}(z,y)  \nonumber  \\
&= \frac{\sqrt{2\epsilon}}{(1-\epsilon^2)^{1/4}(1-\epsilon)^{1/2}\pi^{3/4}}
 e^{-\frac\epsilon{1-\epsilon^2}z^2 - \frac{1+\epsilon}{2(1-\epsilon)}y^2
    +\frac{2\sqrt\epsilon}{1-\epsilon} zy} .  \label{tVD}
\end{align}
Using the isomorpism~$V$, one~can transfer operators on $\FFe$ into those
on~$L^2(\RR)$. This applies, in~particular, also to the Toeplitz
operators~$T_\phi$, $\phi\in L^\infty(\CC)$, on~$\FFe$,
recalled in the Introduction. From the definition
$$ \spr{T_\phi f,g}_{\FFe} = \intC \phi f \overline g \,d\mu_\epsilon,
 \qquad f,g\in\FFe,   $$
using~(\ref{tVA}) one obtains for the transferred operator $V^*T_\phi V$
on~$L^2(\RR)$
\[ V^* T_\phi Vf(x) = \intR f(y) k_\phi(x,y) \, dy   \label{tVC}  \]
where
$$ k_\phi(x,y) = \intC \beta(z,y) \overline{\beta(z,x)} \phi(z)
\,d\mu_\epsilon(z) .   $$
Recall that the Weyl operator on $L^2(\RR)$ with symbol $a(x,\xi)$,
$x,\xi\in\RR$, is~defined~by
$$ W_a f(x) = \intR \intR a(\tfrac{x+y}2,\xi) e^{i(x-y)\xi} f(y)
 \,dy \,\frac{d\xi}{2\pi} ,   $$
where the right-hand side exists as a convergent integral~for, say, $a$~and
$f$ in the Schwarz space, and in general extends to be well-defined as an
oscillatory integral for more general functions or even distributions $f$
on $\RR$ and $a$ on~$\RR^2$; see e.g.~\cite{Foll}. Performing the $\xi$
integration yields
\[ W_a f(x) = \intR \check a(\tfrac{x+y}2,x-y) \, f(y) \, dy, \label{tVB}  \]
where $\check{\;}$ denotes the inverse Fourier transform with respect to the
second variable.

\begin{theorem}
We~have $V^*T_\phi V=W_a$, where
\[ a(x,\xi) = \Big(e^{\frac{1-\epsilon^2}{16\epsilon}\Delta}\phi\Big)
 \Big(\frac{1+\epsilon}{2\sqrt\epsilon}x
    - \frac{1-\epsilon}{2\sqrt\epsilon}i\xi \Big),
  \qquad x,\xi\in\RR.   \label{tVE}   \]
Here $e^{t\Delta}$, $t>0$, denotes the standard heat solution operator
$$ e^{t\Delta}\phi(w) = \frac1{4\pi t} \intC \phi(z) e^{-|z-w|^2/(4t)}\,dt. $$
\end{theorem}

\begin{proof} Comparing (\ref{tVC}) and (\ref{tVB}) we see that
$V^*T_\phi V=W_a$ where $\check a(\frac{x+y}2,x-y)=k_\phi(x,y)$,
or $\check a(s,r)=k_\phi(s+\frac r2,s-\frac r2)$, or
$$ a(s,\eta) = \intR \intC e^{-ir\eta} \beta(z,s-\tfrac r2)
 \overline{\beta(z,s+\tfrac r2)} \, \phi(z) \,d\mu_\epsilon(z) \, dr.  $$
Substituting (\ref{tVD}) for $\beta$ and carrying out the $r$ integration
yields after some calculations
$$ a(s,\eta) = \frac{4\epsilon}{(1-\epsilon^2)\pi} \intC \phi(z)
 e^{-\frac{[(1+\epsilon)s-2\sqrt\epsilon z_1]^2
 +[(1-\epsilon)\eta+2\sqrt\epsilon z_2]^2} {1-\epsilon^2}} \, dz ,
 \qquad z=z_1+i z_2,   $$
which is~(\ref{tVE}).   \end{proof}

Using the standard properties of the Weyl calculus and the last theorem,
it~is possible to recover the semi-classical asymptotics~(\ref{tTE}),
mentioned in the Introduction, for the Toeplitz operators on the Fock
space~$\FFe$. Namely, assume that a symbol $a$ lies in the Shubin
(or~Grossmann-Loupias-Stein) class $GLS^m(\RR^2)$, $m\le0$, that~is,
$$ \sup_{x,\xi\in\RR} \frac{|\partial_x^j\partial_\xi^k a(x,\xi)|}
 {(1+|x|+|\xi|)^{m-j-k}} < \infty  \qquad\forall j,k=0,1,2,\dots,  $$
and let similarly $b\in GLS^n$, $n\le0$. Then it is known that
$W_aW_b=W_c$ for a unique $c\in GLS^{m+n}$, and furthermore
$c=:a\#b$ has asymptotic expansion
\[ (a\#b)(x,\xi) \sim \sum_{k=0}^\infty \frac{(i/2)^k}{k!}
 (\partial_x\partial_\eta-\partial_\xi\partial_y)^k
 a(x,\xi) b(y,\eta) \Big| _{y=x,\eta=\xi},  \label{tVF}  \]
where ``$\sim$'' means that the left-hand side differs from the partial
sum of the first $N$ terms on the right-hand side by an element
from~$GLS^{m+n-2N}$, for all $N=0,1,2,\dots$. Also, for $a\in GLS^m$,
$m\le0$, and any $t>0$, one~has
\[ e^{t\Delta}a \sim \sum_{k=0}^\infty \frac{t^k}{k!} \Delta^k a;
 \label{tVH}  \]
see e.g.~\cite[Theorem~3.1 and \S7.4]{BEY}. In~particular, this holds for
$t=\frac{1-\epsilon^2}{16\epsilon}$; note that then the last formula,
in~addition to holding in the same sense as in~(\ref{tVF}) above,
at~the same time also represents an asymptotic expansion of $e^{t\Delta}a$
as $\epsilon\nearrow1$ in increasing powers of $(1-\epsilon)$.
Introducing momentarily the shorthands
$$ \tau_\epsilon:=e^{\frac{1-\epsilon^2}{16\epsilon}\Delta}, \quad
 \kappa_\epsilon \phi(x,\xi) := \phi(\tfrac{1+\epsilon}{2\sqrt\epsilon}x,
 -\tfrac{1-\epsilon}{2\sqrt\epsilon}\xi),   $$
we~thus get for any $\phi\in GLS^m$, $\psi\in GLS^n$, $m,n\le0$,
$$ V^*(T_\phi T_\psi-T_{\phi\psi})V = W_a  $$
where
\begin{align}
a &= (\kappa_\epsilon \tau_\epsilon\phi)
 \# (\kappa_\epsilon \tau_\epsilon \psi)
 - \kappa_\epsilon \tau_\epsilon(\phi\psi)  \nonumber \\
& \begin{aligned}
&\sim \sum_{j,k,l=0}^\infty \frac{(i/2)^k t^{j+l}}{j!k!l!}
 (\partial_{x,\phi}\partial_{\xi,\psi}-\partial_{\xi,\phi}\partial_{x,\psi})^k
 (\kappa_\epsilon \Delta^j\phi) (\kappa_\epsilon \Delta^l\psi)   \\
&\hskip2em - \sum_{k=0}^\infty \frac{t^k}{k!}
 \kappa_\epsilon \Delta^k(\phi\psi) ,  \end{aligned}   \label{tVG}
\end{align}
where the subscripts in $\partial_{x,\phi}$, $\partial_{\xi,\psi}$,
$\partial_{\xi,\phi}$, $\partial_{x,\psi}$ indicate which of the
functions $\partial_x$ or $\partial_\xi$ applies~to,
and $t=\frac{1-\epsilon^2}{16\epsilon}$.
Observe that each $\partial_\xi$ picks from $\kappa_\epsilon$ a~factor
of~$(1-\epsilon)$, so~that the last ``$\sim$'' again, in~addition to
holding in the same sense as in~(\ref{tVF}), is~also an asymptotic
expansion in descending powers of $(1-\epsilon)$ as $\epsilon\nearrow1$.
(Note that $\kappa_\epsilon$ evidently maps each $GLS^m$ into itself.)

Since $\kappa_\epsilon(\phi\psi)=(\kappa_\epsilon\phi)(\kappa_\epsilon\psi)$,
the~top order terms in the two sums in (\ref{tVG}) cancel~out. The~terms with
$j+k+l=1$ in the first sum and the term $k=1$ of the second sum combine into
\begin{align*}
& t\kappa_\epsilon\Delta\phi \cdot \kappa_\epsilon\psi
 + t\kappa_\epsilon\phi \cdot \kappa_\epsilon\Delta\psi
 + \tfrac i2 (\partial_x \kappa_\epsilon\phi
              \cdot \partial_\xi \kappa_\epsilon\psi
            - \partial_\xi \kappa_\epsilon\phi
              \cdot \partial_x \kappa_\epsilon\psi)
 - t\kappa_\epsilon\Delta(\phi\psi)  \\
&\hskip4em =
 -2t\kappa_\epsilon [(\partial_x\phi-i\partial_\xi\phi)
   (\partial_x\psi+i\partial_\xi\psi)] .
\end{align*}
Thus, appealing one more time to~(\ref{tVH}),
$$ a = -2t\tau_\epsilon \kappa_\epsilon [(\partial_x\phi-i\partial_\xi\phi)
   (\partial_x\psi+i\partial_\xi\psi)] + b  $$
where $b\in GLS^{m+n-4}$ and also $b=O(t^2)$ as $t\searrow0$,
i.e.~$\epsilon\nearrow1$. Back~on the level of~$T_\phi$, this amount~to
$$ T_\phi T_\psi-T_{\phi\psi} = -2t T_{(\partial_x\phi-i\partial_\xi\phi)
   (\partial_x\psi+i\partial_\xi\psi)} + O(t^2),   $$
and upon interchanging $\phi,\psi$ and subtracting,
$$ T_\phi T_\psi - T_\psi T_\phi = \frac{ih}{2\pi} T_{\{f,g\}} + O(h^2)  $$
with the Poisson bracket $\{\phi,\psi\}=\partial_\xi\phi \partial_x\psi -
\partial_x\phi \partial_\xi\psi$ and Planck's constant
$$ h = \frac{1-\epsilon^2}{2\epsilon}\pi,   $$
thus recovering~(\ref{tTA}).

Using the further terms in~(\ref{tVG}), it~is plain how to recover the
complete semiclassical expansion (\ref{tTE}) as~well.

We~conclude by remarking that analogously to~(\ref{tVZ}), we~also have the
unitary map
$$ U: f \longmapsto \sum_n \Big(\frac{2\epsilon}{1-\epsilon}\Big)^{(n+1)/2}
\Big\langle f,\frac{z^n}{\sqrt{n!\pi}} \Big\rangle E_n(z)   $$
in $\FFe$ sending the standard monomial orthonormal basis $\{(\frac{2\epsilon}
{(1-\epsilon)})^{(n+1)/2}(n!\pi)^{-1/2}z^n\}_n$ into~$E_n$; thus $VU^*$ is the
usual Bargmann transform of $L^2(\RR)$ onto~$\FFe$. Explicitly,
$$ U f(w) = \frac{2\epsilon}{(1-\epsilon^2)^{3/4}\pi}
 \intC f(z) e^{\frac{2\epsilon w\oz}{\sqrt{1-\epsilon^2}} \,
  -\frac{\epsilon^2}{1-\epsilon^2}(\oz^2-w^2)} \, d\mu_\epsilon(w) ,  $$
which can also be written, using the reproducing kernel property,
$$ U = T_{1/\psi} \delta_{1/\sqrt{1-\epsilon^2}} T^*_{\psi},
 \qquad \psi(z) = e^{-\epsilon^2 z^2/(1-\epsilon^2)},  $$
as~a~product of two Toeplitz operators and a dilation.

\medskip  
Let us conclude by making a conjecture. The lack of an obvious physical
interpretation for the results obtained by our ``semiclassical analysis''
above, might be a reflection of the fact that an underlying localization
property of the quantized system is absent here. As is well known, when the
reproducing kernel is a subspace of an $L^2$-space, there exists a family of
localization operators and a positive operator valued measure which define the
localization properties of the quantum system in $\Omega$. In the absence of
such an ambient space, no such measure is available and hence no obvious sense
in which the quantum system is localized in $\Omega$. 
In~any case, the authors find the application of orthogonal polynomials
to the construction of the associated reproducing kernel spaces and operators
on~them a rather charming figment of complex analysis, and hope very much to
have at least partly conveyed this feeling to the reader as well.


\begin{thebibliography}{99}

\bibitem{AliKr} S. T. Ali, K.~Gorska, A.~Horzela, F.H.~Szafraniec:
{\it Squeezed states and Hermite polynomials in a complex variable,\/}
arxiv:1308.4730, J. Math. Phys., to appear.

\bibitem{alidoeb} S. T. Ali, H.-D. Doebner: {\it Ordering problem in quantum mechanics: Prime quantization and its physical interpretation,\/} Phys. Rev. A {\bf 41} (1990), 1199--1210.

\bibitem{Aro} N. Aronszajn: {\it Theory of reproducing kernels\/},
Trans. Amer. Math. Soc. {\bf 68} (1950), 337--404.

\bibitem{Askey} G.E. Andrews, R. Askey, R. Roy, {\it Special functions,\/}
Cambridge University Press, Cambridge, 1999.

\bibitem{AUrealcov} J. Arazy, H. Upmeier: {\it Covariant symbolic calculi on
real symmetric domains,\/} Singular Integral Operators, Factorization and
Applications, Oper.~Theory Adv.~Appl., vol.~142, Birkh\"auser, Basel, 2003,
pp.~1--27.

\bibitem{BeQ} F.A. Berezin: {\it Quantization,\/} Math. USSR Izvestiya
{\bf8} (1974), 1109--1163.

\bibitem{Blaschke} P. Blaschke: {\it Berezin transform on harmonic Bergman
spaces on the unit ball in~$\RR^n$,\/} preprint, 2012.

\bibitem{BEY} H. Bommier-Hato, M. Engli\v s, E.-H. Youssfi:
{\it Dixmier trace and the Fock space,\/} Bull. Sci. Math., to~appear
(http://dx.doi.org/10.1016/j.bulsci.2013.04.009).

\bibitem{BMS} M. Bordemann, E. Meinrenken, M. Schlichenmaier:
{\it Toeplitz quantization of K\"ahler manifolds and $gl(n)$,
$n\to\infty$ limits,\/} Comm. Math. Phys. {\bf 165} (1994), 281--296.

\bibitem{BwUribe} D. Borthwick, A. Uribe: {\it Nearly K\"ahlerian embeddings
of symplectic manifolds,\/} Asian J. Math. {\bf 4} (2000), 599--620.

\bibitem{DaiLiuMa} X. Dai, K. Liu and X. Ma, {\it On the asymptotic expansion
of Bergman kernel,\/} J. Differential Geom. {\bf 72} (2006), 1--41.

\bibitem{DougKle} M. Douglas, S. Klevtsov, {\it Bergman kernel from path
integral,\/} Comm. Math. Phys. {\bf 293} (2010), 205--230.

\bibitem{Ecmp} M. Engli\v s: {\it Weighted Bergman kernels and quantization,\/}
Comm. Math. Phys. {\bf 227} (2002), 211--241.

\bibitem{EFock} M. Engli\v s: {\it Berezin transform on the harmonic Fock
space,\/} J.~Math. Anal. Appl. {\bf 367} (2010), 75--97.

\bibitem{Foll} G.B. Folland, {\it Harmonic analysis in phase space,\/}
Princeton University Press, 1989.

\bibitem{Gammelg} N. Gammelgaard: {\it A universal formula for deformation
quantization on K\"ahler manifolds,\/} arXiv:1005.2094.

\bibitem{Hrm} L. H\"ormander: {\it The analysis of linear partial
differential operators, vol.~I,\/} Grund\-leh\-ren der mathematischen
Wissenschaften, vol.~256, Springer-Verlag, Berlin - Heidelberg - New York
- Tokyo, 1985.

\bibitem{horszaf} A. Horzela, F.H. Szafraniec: {\it A measure-free approach to coherent states,\/} J. Phys. A: Math. Theor. {\bf 45} (2012), 244018  doi:10.1088/1751-8113/45/24/244018.

\bibitem{Jahn} J. Jahn: {\it On~the asymptotic expansion of Berezin
transform on the half-space,\/} J.~Math. Anal. Appl. {\bf 405} (2013),
720--730.

\bibitem{KarbSchl} A.V. Karabegov, M. Schlichenmaier: {\it Identification of
Berezin-Toeplitz deformation quantization,\/} J.~reine angew. Math. {\bf 540}
(2001), 49--76.

\bibitem{MaMa} X. Ma, G. Marinescu: {\it The first coefficients of the
asymptotic expansion of the Bergman kernel of the $\operatorname{spin}^c$
Dirac operator,\/} Internat. J. Math. {\bf 17} (2006), 737--759.

\bibitem{odzhor} A. Odzijewicz, M. Horowski: {\it Positive kernels and
quantization,\/} J. Geom. Physics {\bf 63} (2013), 80--98. 

\bibitem{prugov} E. Prugove\v cki: {\it Consistent formulation of relativistic
dynamics for massive spin-zero particles in external fields,\/} Phys. Rev. D
{\bf 18} (1978), 3655--3673. 

\bibitem{ReshTakh} N. Reshetikhin, L. Takhtajan: {\it Deformation quantization
of K\"ahler manifolds,\/} L.D. Faddeev's Seminar on Mathematical Physics,
Amer. Math. Soc. Transl. Ser. 2, Vol. 201, AMS, Providence, 2000, pp.~257--276.

\bibitem{ShiffZeld} B. Shiffman, S. Zelditch: {\it Asymptotics of almost
holomorphic sections of ample line bundles on symplectic manifolds,\/}
J.~reine Angew. Math. {\bf 544} (2002), 181--222.

\bibitem{HXua} H.~Xu: {\it An explicit formula for the Berezin star product,\/}
Lett. Math. Phys. {\bf 101} (2012), 239--264.

\bibitem{HXub} H.~Xu: {\it On a graph theoretic formula of Gammelgaard for
Berezin-Toeplitz quantization,\/} Lett. Math. Phys. {\bf 103} (2013), 145--169.


\end{thebibliography}
\end{document}